\documentclass[12pt]{article} 

\usepackage[ngerman,english]{babel}
\useshorthands{"}
\addto\extrasenglish{\languageshorthands{ngerman}}
\usepackage[utf8]{inputenc}
\usepackage[T1]{fontenc}
\usepackage{amsmath}
\usepackage{graphicx}
\usepackage{amssymb}
\usepackage{hyperref}
\usepackage{amsmath}
\usepackage{mathrsfs}
\usepackage{color}
\usepackage{subfig}
\usepackage{amsthm}

\newtheorem{definition}{Definition}[section]

\newtheorem{lemma}{Lemma}[section]

\theoremstyle{definition}

\renewcommand{\theta}{\vartheta}

\title{{\Large\bf Proof of a Generalized Ryu-Takayanagi Conjecture}}
\author{{\bf Artem Averin$^{\textrm{a}}$\footnote{artem.averin@campus.lmu.de}}}

\begin{document}

\maketitle

\centerline{\it $^{\textrm{a}}$ Arnold--Sommerfeld--Center for Theoretical Physics,}
\centerline{\it Ludwig--Maximilians--Universit\"at, 80333 M\"unchen, Germany}

\vskip1cm
\begin{abstract}
{{
We derive a generalized version of the Ryu-Takayanagi formula for the entanglement entropy in arbitrary diffeomorphism invariant field theories. We use a recent framework which expresses the measurable quantities of a quantum theory as a weighted sum over paths in the theory's phase space. If this framework is applied to a field theory on a spacetime foliated by a hypersurface $\Sigma,$ the choice of a codimension-2 surface $B$ without boundary contained in $\Sigma$ specifies a submanifold in the phase space. For diffeomorphism invariant field theories, a functional integral expression for their density matrices was recently given and then used to derive bounds on phase space volumes in the considered submanifold associated to $B.$ These bounds formalize the gravitational entropy bound. Here, we present an implication of this derivation in that we show the obtained functional integral expression for density matrices to be naturally suited for the replica trick. Correspondingly, we prove a functional integral expression for the associated entanglement entropies and derive a practical prescription for their evaluation to leading order and beyond. An important novelty of our approach is the contact to phase space. This allows us both to obtain a prescription for entanglement entropies in arbitrary diffeomorphism invariant field theories not necessarily possessing a holographic dual as well as to use entanglement entropies to study their phase space structure. In the case of the bulk-boundary correspondence, our prescription consistently reproduces and hence provides a natural and independent proof of the Ryu-Takayanagi formula as well as its various generalizations. These include the covariant holographic entanglement entropy proposal, Dong's proposal for higher-derivative gravity as well as the quantum extremal surface prescription.             
}}
\end{abstract}


\newpage

\setcounter{tocdepth}{2}
\tableofcontents
\break

\section{Geometry of Possibilities}
\label{Kapitel 1}

The measurable quantities of a quantum theory can be represented as weighted sums of the possibilities of how the system can evolve in principle. The set of possibilities together with the associated weight factors specifies hereby the theory under consideration.

A natural question is then which possibilities can be resolved or distinguished by the measurable quantities (and hence correspond to quantum mechanically distinct states)? We argue that, if Nature is indeed completely described by a quantum theory, the mentioned question addresses in a unified way various open problems in physics, in particular: Why can black holes only resolve a finite number of microstates as given by their entropy (and how to describe them)? Why are colored states not resolvable at low energies in quantum chromodynamics and why is there a mass gap in the spectrum?

The question raised above is in general related to how sensitive the measurable quantities (given by a weighted sum over the possibilities) are to the different possibilities. A set of possibilities can be of high or low influence to this sum. Hence, a given set of possibilities is motivated to be thought of as being assigned a volume which is a measure of the quantum mechanical sensitivity of observable quantities on this particular set of possibilities. The set of possibilities then naturally possesses a geometric structure. We learn that the study of this geometric structure corresponds to the questions raised above. 

A precise mathematical framework for the quantum mechanical sum of possibilities discussed qualitatively above was given in \cite{Averin:2024its}. In this framework, measurable quantities of a quantum theory are given by a weighted sum of paths in the theory's phase space. According to our discussion above, we are motivated to study the geometric structure of this phase space for a given theory. However, phase spaces are technically speaking symplectic manifolds which are known to be locally indistinguishable. Hence, differential methods are not helpful in analyzing their geometries. Instead, we want to consider weighted sums over paths in those phase spaces which allow to draw conclusions about the phase space structure and geometry. In this sense, we are doing ``integral geometry.'' The precise form of the considered sums is hereby motivated quantum mechanically.

This geometric viewpoint on quantum mechanics as the sum of possibilities indeed turns out to be insightful. In \cite{Averin:2024its}, the notion of partitioning this sum by forming suited bundles of possibilities (termed possifolds) was introduced. For a field theory, it was shown that there is a natural way to form these bundles. Technically speaking, for a field theory on a spacetime foliated by a hypersurface $\Sigma,$ the possifold flow naturally associates to the choice of a codimension-2 surface $B \subseteq \Sigma$ a submanifold in the theory's phase space. 

Using the discussed integral geometric reasoning, it was shown in \cite{Averin:2024yeo} that this submanifold is for diffeomorphism invariant field theories restricted to obey a bound on phase space volumes. This formalizes and derives what was heuristically expected to be the gravitational entropy bound.

The key point of the derivation in \cite{Averin:2024yeo} is to consider a certain sum of paths in the theory's phase space that quantum mechanically is interpreted as giving the density matrix elements of a state where the reduction of Hilbert-space is performed by using the submanifold associated to $B$ with the possifold flow. The intimate interplay between the commutativity of this sum together with diffeomorphism invariance allows to rewrite those density matrix elements as a sum over paths in the submanifold associated to $B.$ From this functional integral expression of the density matrix, the gravitational entropy bound then directly follows. 

Here, we want to introduce a further integral geometric tool for the study of phase space structures. In fact, our results here are a direct implication of the proof of the gravitational entropy bound in \cite{Averin:2024yeo}. There, the integral geometric tool was to obtain a suited functional integral expression for density matrices. Here, we will focus on the related von Neumann entropies which in this context have the meaning of an entanglement entropy. 

Therefore, our main result here will be to derive a functional integral expression for entanglement entropies in arbitrary diffeomorphism invariant field theories and to provide a prescription for their evaluation to leading order and beyond. 

This result is useful for several reasons. 

First, gravitational entanglement entropies were discussed in the literature in the context of the bulk-boundary correspondence where the bulk gravitational theory possesses a holographic boundary dual (see \cite{Rangamani:2016dms} for a review). The Ryu-Takayanagi conjecture made in \cite{Ryu:2006bv,Ryu:2006ef} proposes a prescription for the evaluation of certain entanglement entropies in the context of the bulk-boundary correspondence. The conjecture was generalized in various directions such as to the Hubeny-Rangamany-Takayanagi covariant holographic entanglement entropy proposal \cite{Hubeny:2007xt}, Dong's proposal on higher-derivative gravity \cite{Dong:2013qoa}, beyond leading order proposals made in \cite{Faulkner:2013ana} and the quantum extremal surface prescription \cite{Engelhardt:2014gca}. Proofs of the individual proposals were given in \cite{Lewkowycz:2013nqa,Dong:2017xht} (see also the review \cite{Rangamani:2016dms} and references therein).

Our result generalizes the various versions of the Ryu-Takayanagi conjecture to the general case of diffeomorphism invariant field theories not necessarily possessing a holographic dual. If so, it consistently reproduces the known results. Thereby, our proof provides an independent derivation which explains the physical and mathematical origin of the Ryu-Takayanagi type prescriptions. 

Second, our derived functional integral expression for gravitational entanglement entropies is given by a sum over paths in a submanifold of phase space determined by the mentioned concept of the possifold flow. Hence, in contrast to the known Ryu-Takayanagi type prescriptions, our gravitational entanglement entropy prescription always makes contact to the phase space. 

On the one hand, this contact to phase space is the reason why we are able to give a sensible definition of gravitational entanglement entropies in a non-holographic setting. 

On the other hand, the contact to phase space allows us using entanglement entropies as a tool of integral geometry as we aimed at.

It is precisely an application of our results presented here to apply this mentioned tool involving entanglement entropy to situations with the presence of black holes. We would like then to infer which part of the phase space in gravitational theories is responsible for their entanglement entropy and consequently for their microstates? However, we want to devote this analysis its own place. Our purpose in this paper is to develop the necessary tool in the form of a derivation of a generalized version of the Ryu-Takayanagi prescription. 

The paper is organized as follows. In chapter \ref{Kapitel 2}, we derive a general functional integral expression for entanglement entropies in arbitrary diffeomorphism invariant field theories. We derive a practical prescription for their evaluation to leading order and beyond in chapter \ref{Kapitel 3}. In chapter \ref{Kapitel 4}, we show that our derived prescription reproduces in the special case of the bulk-boundary correspondence the Ryu-Takayanagi formula and its various known generalizations. We give a discussion and outlook in chapter \ref{Kapitel 5}. 

We set Planck's constant $\hbar$ to one although we restore it occasionally in order to emphasize the entrance of quantum mechanics.        

\section{Derivation}
\label{Kapitel 2}

We use the formalism and language developed in \cite{Averin:2024its}.

We want to derive an expression for the entanglement entropy in diffeomorphism invariant field theories. A natural first step to do so would be to give an expression for density matrices in such theories. In fact, this turns out to be the main step and it was already done in \cite{Averin:2024yeo}. Hence, we can literally follow chapter $2$ of \cite{Averin:2024yeo} until Lemma $2.3.$ Especially, we adopt the same notations and conventions of \cite{Averin:2024yeo}.

That is, we consider a diffeomorphism invariant field theory $(\Gamma,\Theta,H)$ on spacetime $M=\mathbb{R} \times \Sigma$ and associated Hilbert-space $\mathcal{H}.$ We choose a state of the form

\begin{equation}
| \Psi \rangle = e^{HT_E} | \chi \rangle
\label{1}
\end{equation}

for fixed $T_E \in \mathbb{R}$ and a position-eigenstate $| \chi \rangle.$ We factorize $\mathcal{H}$ by considering the possifold-flow $\partial \Sigma \to B$ as introduced in \cite{Averin:2024its} for a codimension-2 surface $B$ bounding a subset $\Sigma_1 = \Sigma(B) \subseteq \Sigma$ and $\Sigma_2$ being its complement $\Sigma_2 = \Sigma \setminus \Sigma_1.$ As explained in \cite{Averin:2024its}, a position-eigenstate is then of the form

\begin{equation}
|q \rangle = |q_1 \rangle |q_2 \rangle
\label{2}
\end{equation}

with $q_1$ and $q_2$ being generalized coordinates on $\Sigma_1$ and $\Sigma_2,$ respectively. We use the factorization \eqref{2} in order to reduce the density matrix associated to the state \eqref{1}. The reduced density matrix is then 

\begin{equation}
\rho = \int \mathcal{D}q_2 \langle q_2 | \left( | \Psi \rangle \langle \Psi | \right) |q_2 \rangle.
\label{3}
\end{equation}

For the details on the discussion of the last paragraph, we refer to the original references \cite{Averin:2024its,Averin:2024yeo}. However, as in \cite{Averin:2024its,Averin:2024yeo}, we again want to emphasize that it is the field theory nature of the considered theory $(\Gamma,\Theta,H)$ which leads to a natural factorization of the associated Hilbert-space $\mathcal{H}.$ The particular choice of a factorization \eqref{2} is naturally tied to the choice of a codimension-2 surface $B.$ This factorization is then used to define the reduced density matrix \eqref{3} of the state \eqref{1}. In other words, for field theories, there is a natural way to reduce density matrices. If, in addition, the field theory is diffeomorphism invariant, the reduced density matrix elements possess a natural functional integral expression as was shown in \cite{Averin:2024yeo}. Using \eqref{2}, the matrix elements of \eqref{3} for position-eigenstates $|q_{1,A} \rangle$ and $|q_{1,B} \rangle$ are given by

\begin{equation} \label{4}
\langle q_{1,A} | \rho |q_{1,B} \rangle 
= \int_{q (\alpha = 0)=q_{1,B} \atop q(\alpha=\pi, t=\pm iT_E)=\chi}^{q(\alpha=2\pi)=q_{1,A}} \operatorname{Vol}^{(B)}(\alpha) e^{i\int_{0}^{2\pi}d \alpha \left( \Theta^{(B)}_{\Phi(\alpha)} \left[\dot{\Phi} (\alpha) \right] - G[\Phi(\alpha)] \right)}.
\end{equation}

The result was shown in Lemma $2.3$ in \cite{Averin:2024yeo}. $G$ is a generator on $(\Gamma^{(B)},\Theta^{(B)})$ of a certain diffeomorphism. The explicit form of $G$ is given in \cite{Averin:2024yeo}. 

With the equality \eqref{4}, we have complete knowledge of the density matrix $\rho.$ From this, we want to find an expression for the entanglement entropy

\begin{equation}
S = -\operatorname{tr} (\rho \ln (\rho)).
\label{5}
\end{equation}

A well-known way to do so is by what is known as the replica trick (see, for instance, \cite{Rangamani:2016dms}). It consists of computing $\rho^n$ from $\rho$ for integer powers $n \in \mathbb{N}.$ The entropy \eqref{5} is then deduced by analytic continuation in $n.$ 

In our situation, the replica trick realizes very naturally. For $n \in \mathbb{N},$ we obtain from \eqref{4} by matrix multiplication the expression for $\rho^n$ as

\begin{equation} \label{6}
\langle q_{1,A} | \rho^n |q_{1,B} \rangle 
= \int_{q (\alpha = 0)=q_{1,B} \atop q(\alpha \in \pi(2\mathbb{Z}+1), t=\pm iT_E)=\chi}^{q(\alpha=2\pi n)=q_{1,A}} \operatorname{Vol}^{(B)}(\alpha) e^{i\int_{0}^{2\pi n}d \alpha \left( \Theta^{(B)}_{\Phi(\alpha)} \left[\dot{\Phi} (\alpha) \right] - G[\Phi(\alpha)] \right)}.
\end{equation}

Both sides of \eqref{6} can be analytically continued in $n.$ Because of Carlson's theorem, the equality \eqref{6} stays valid for real values of $n \in \mathbb{R}.$

In the context of the replica trick, $n$ is known as the R\'{e}nyi-index. In our situation, it has a vivid meaning. As mentioned above, $G$ generates a diffeomorphism constructed in \cite{Averin:2024yeo}. Near $B,$ the diffeomorphism looks in a locally flat coordinate system like a rotation in the normal plane to $B.$ In \eqref{6}, the R\'{e}nyi-index $n$ counts the number of full rotations in this plane. The expression \eqref{6} naturally allows a continuation of the R\'{e}nyi-index $n$ to non-integer values. We refer to \cite{Averin:2024yeo} for more details on the diffeomorphism generated by $G$ where it was discussed and visualized in great detail. 

Having now a functional integral expression \eqref{6} for $\rho^n,$ we would like to proceed the replica trick in order to compute the associated entropy. For this purpose, we want to introduce an abbreviation of the integration limits in expressions like \eqref{6}. This motivates the following

\begin{definition} \label{Definition 2.1}
A path $\Phi \colon [0,2\pi n] \to \Gamma^{(B)}$ in $\Gamma^{(B)}$ with the properties

\begin{equation} \label{7}
\begin{split}
q (\alpha = 0)&=q_{1,B} \\
q(\alpha=2\pi n)&=q_{1,A} \\
q(\alpha \in \pi(2\mathbb{Z}+1), t=\pm iT_E)&=\chi
\end{split}
\end{equation}

is called a \emph{replica spacetime}. The set of all such paths is called $R_n.$
\end{definition}

Several remarks are in order.

With Definition \ref{Definition 2.1}, \eqref{6} expresses the matrix elements of $\rho^n$ as a sum over replica spacetimes. More generally, we will use the term replica spacetime for such an abbreviation to apply. The Definition \ref{Definition 2.1} has to be understood accordingly. For instance, if we were interested in $\operatorname{tr}(\rho^n),$ we would drop the first two equations in \eqref{7} and instead require the path in Definition \ref{Definition 2.1} to be closed. Again, by \eqref{6}, $\operatorname{tr}(\rho^n)$ is then given by a summation of replica spacetimes. 

More precisely,

\begin{equation}
\operatorname{tr} (\rho^n)
= \int_{\Phi \in R_n} \operatorname{Vol}^{(B)}(\alpha) e^{i\int_{0}^{2\pi n}d \alpha \left( \Theta^{(B)}_{\Phi(\alpha)} \left[\dot{\Phi} (\alpha) \right] - G[\Phi(\alpha)] \right)}
\label{8}
\end{equation}

where $R_n$ is the appropriate mentioned definition of a replica spacetime. 

Since we are considering diffeomorphism invariant field theories, the term replica \emph{spacetime} is adequate. This is because diffeomorphism invariant field theories possess a Lagrangian that can always be written as depending on a spacetime metric (see the discussion in \cite{Averin:2024yeo} and references therein). In fact, since $\Gamma^{(B)}$ is the phase space of the canonical coordinates on $\Sigma(B) \subseteq \Sigma,$ we can view a path $\Phi \colon [0,2\pi n] \to \Gamma^{(B)}$ as realizing a spacetime $[0, 2\pi n] \times \Sigma(B)$ foliated by $\Sigma(B)$ along the time $[0, 2\pi n].$ In this way, we understand a replica spacetime as a \emph{spacetime}. 

The replica trick demands us to compute the derivative of $\rho^n$ with respect to $n$ in order to finally get an expression for the entanglement entropy. We use the expression for $\rho^n$ for non-integer $n$ found in \eqref{6} together with the notation introduced in Definition \ref{Definition 2.1} to perform the required derivative in the next

\begin{lemma} \label{Lemma 2.1}
For the expression \eqref{6}, one has

\begin{equation} \label{9}
\begin{split}
&-\frac{d}{dn} \langle q_{1,A} | \rho^n |q_{1,B} \rangle \\
=& \int_{\Phi \in R_n} \operatorname{Vol}^{(B)}(\alpha) K[\Phi(\alpha=2\pi n)] e^{i\int_{0}^{2\pi n}d \alpha \left( \Theta^{(B)}_{\Phi(\alpha)} \left[\dot{\Phi} (\alpha) \right] - G[\Phi(\alpha)] \right)}
\end{split}
\end{equation}

with

\begin{equation}
K=2\pi i G
\label{10}
\end{equation}

being a generator on $(\Gamma^{(B)},\Theta^{(B)}).$
\end{lemma}

\begin{proof}
For $n \in \mathbb{Z},$ let be $dn > 0.$ We have according to \eqref{6} using Definition \ref{Definition 2.1}

\begin{equation} \label{11}
\begin{split}
& \langle q_{1,A} | \rho^{n+dn} |q_{1,B} \rangle - \langle q_{1,A} | \rho^n |q_{1,B} \rangle \\
=& \int_{\Phi \in R_{n+dn}} \operatorname{Vol}^{(B)}(\alpha) e^{i\int_{0}^{2\pi (n+dn)}d \alpha \left( \Theta^{(B)}_{\Phi(\alpha)} \left[\dot{\Phi} (\alpha) \right] - G[\Phi(\alpha)] \right)} \\
-& \int_{\Phi \in R_n} \operatorname{Vol}^{(B)}(\alpha) e^{i\int_{0}^{2\pi n}d \alpha \left( \Theta^{(B)}_{\Phi(\alpha)} \left[\dot{\Phi} (\alpha) \right] - G[\Phi(\alpha)] \right)}.
\end{split}
\end{equation}

We define the mapping

\begin{equation}
R_n \to R_{n+dn}, \Phi \mapsto \tilde{\Phi}
\label{12}
\end{equation}

via 

\begin{equation} \label{13}
\tilde{\Phi}(\alpha) :=
\begin{cases}
  \Phi(\alpha),  & \text{for } \alpha \in [0,2\pi(n-dn)]\\
  \Phi(\tilde{\alpha}(\alpha)), & \text{for } \alpha \in [2\pi (n-dn),2\pi (n+dn)].
\end{cases}
\end{equation}

Hereby,

\begin{equation}
\tilde{\alpha} \colon [2\pi (n-dn),2\pi (n+dn)] \to [2\pi (n-dn), 2\pi n]
\label{14}
\end{equation}

is the linear function uniquely determined by requiring

\begin{equation} \label{15}
\begin{split}
\tilde{\alpha} (2\pi (n-dn)) &= 2\pi (n-dn) \\
\tilde{\alpha} (2\pi (n+dn)) &= 2\pi n.
\end{split}
\end{equation}

In total, the mapping \eqref{12} is a change of the time label for a path in $\Gamma^{(B)}.$ As such, $\Phi \mapsto \tilde{\Phi}$ is a bijection which keeps the measure of the functional integration in \eqref{11} invariant. Therefore, we can conclude

\begin{equation} \label{16}
\begin{split}
& \langle q_{1,A} | \rho^{n+dn} |q_{1,B} \rangle - \langle q_{1,A} | \rho^n |q_{1,B} \rangle \\
=& \int_{\Phi \in R_{n}} \operatorname{Vol}^{(B)}(\alpha) \left( e^{i\int_{0}^{2\pi (n+dn)}d \alpha \left( \Theta^{(B)}_{\tilde{\Phi}(\alpha)} \left[\dot{\tilde{\Phi}} (\alpha) \right] - G[\tilde{\Phi}(\alpha)] \right)} \right. \\
-& \left. e^{i\int_{0}^{2\pi n}d \alpha \left( \Theta^{(B)}_{\Phi(\alpha)} \left[\dot{\Phi} (\alpha) \right] - G[\Phi(\alpha)] \right)} \right) \\
=& \int_{\Phi \in R_n} \operatorname{Vol}^{(B)}(\alpha) e^{i\int_{0}^{2\pi n}d \alpha \left( \Theta^{(B)}_{\Phi(\alpha)} \left[\dot{\Phi} (\alpha) \right] - G[\Phi(\alpha)] \right)} \cdot \\
&i \left( \int_{0}^{2\pi (n+dn)}d \alpha \left( \Theta^{(B)}_{\tilde{\Phi}(\alpha)} \left[\dot{\tilde{\Phi}} (\alpha) \right] - G[\tilde{\Phi}(\alpha)] \right) \right. \\
&- \left. \int_{0}^{2\pi n}d \alpha \left( \Theta^{(B)}_{\Phi(\alpha)} \left[\dot{\Phi} (\alpha) \right] - G[\Phi(\alpha)] \right) \right) + O(dn^2). 
\end{split}
\end{equation}

We have

\begin{equation}
\int_{0}^{2\pi (n+dn)}d \alpha \Theta^{(B)}_{\tilde{\Phi}(\alpha)} \left[\dot{\tilde{\Phi}}(\alpha)\right] - 
\int_{0}^{2\pi n}d \alpha \Theta^{(B)}_{\Phi(\alpha)} \left[\dot{\Phi} (\alpha) \right] 
=\int_{\tilde{\Phi}} \Theta^{(B)} - \int_{\Phi} \Theta^{(B)} = 0
\label{17}
\end{equation}

because both $\Phi$ and $\tilde{\Phi}$ trace out the same path in $\Gamma^{(B)}.$ Furthermore,

\begin{equation}
\int_{0}^{2\pi (n+dn)}d \alpha  G[\tilde{\Phi}(\alpha)] - \int_{0}^{2\pi n}d \alpha G[\Phi(\alpha)]
=2\pi dn \cdot G[\Phi(\alpha=2\pi n)] + O(dn^2).
\label{18}
\end{equation}

Using \eqref{17} and \eqref{18}, we obtain from \eqref{16}

\begin{equation} \label{19}
\begin{split}
& \langle q_{1,A} | \rho^{n+dn} |q_{1,B} \rangle - \langle q_{1,A} | \rho^n |q_{1,B} \rangle \\
=& \int_{\Phi \in R_n} \operatorname{Vol}^{(B)}(\alpha) e^{i\int_{0}^{2\pi n}d \alpha \left( \Theta^{(B)}_{\Phi(\alpha)} \left[\dot{\Phi} (\alpha) \right] - G[\Phi(\alpha)] \right)} \cdot \\
&i(-2\pi dn G[\Phi(\alpha=2\pi n)]) + O(dn^2).
\end{split}
\end{equation}

With the definition \eqref{10} of the generator $K,$ \eqref{9} follows from \eqref{19}.
\end{proof}

With Lemma \ref{Lemma 2.1} at hand, we can obtain a functional integral expression for the entanglement entropy. For a replica spacetime $\Phi \in R_n,$ we define the action functional

\begin{equation}
I_n [\Phi] := -i\int_{0}^{2\pi n}d \alpha \left( \Theta^{(B)}_{\Phi(\alpha)} \left[\dot{\Phi} (\alpha) \right] - G[\Phi(\alpha)] \right)
\label{20}
\end{equation}

such that we can write the partition functions \eqref{8} as

\begin{equation}
Z_n := \operatorname{tr} (\rho^n)
= \int_{\Phi \in R_n} \operatorname{Vol}^{(B)}(\alpha) e^{-I_n [\Phi]}.
\label{21}
\end{equation}

We remind that in \eqref{21} the appropriate definition of $R_n,$ i.e. the same as in \eqref{8}, has to be used. Lemma \ref{Lemma 2.1} determines the derivative of the partition functions $Z_n$ with respect to the R\'{e}nyi-index as

\begin{equation}
-\frac{d Z_n}{dn} = \int_{\Phi \in R_n} \operatorname{Vol}^{(B)}(\alpha) K[\Phi(\alpha=2\pi n)] e^{-I_n [\Phi]}.
\label{22}
\end{equation}

The expressions \eqref{21} and \eqref{22} are our main results. The reason is that they give explicit functional integral expressions for what is known \cite{Rangamani:2016dms} as the modular entropy

\begin{equation}
\tilde{S}_n := \left(1 - n\frac{d}{dn} \right) \ln (Z_n).
\label{23}
\end{equation}

The modular entropy for $n=1$ is the searched entanglement entropy \eqref{5}

\begin{equation}
S = \tilde{S}_1.
\label{24}
\end{equation}

This is easily verified using the definition in \eqref{21} together with the normalization

\begin{equation}
Z_1 = \operatorname{tr} (\rho) = 1
\label{25}
\end{equation}

of the density matrix.

To summarize, we have found in this chapter an explicit functional integral expression for the entanglement entropy (or more generally the modular entropy) that applies in any diffeomorphism invariant field theory. We have seen that the origin of this fixing of the modular entropy lies in both the covariance of the functional integral - an intrinsic quantum mechanical property reflecting the superposition principle - and in diffeomorphism invariance. As is most often the case in quantum field theory, we cannot expect in being able to fully sum up the functional integration in practical computations of
the modular entropies. Instead, we need to resort to approximations. In the next chapter, we want to find such proper approximations of the functional integral expressions for the entanglement entropy found in this chapter. 

\section{Evaluation}
\label{Kapitel 3}

In this chapter, we would like to ask how to find the entanglement entropy \eqref{24} explicitly. Having found explicit expressions in terms of the functional integrations \eqref{21} and \eqref{22}, we ask how to evaluate them in practice?

Before discussing the evaluation of \eqref{21} and \eqref{22}, let us briefly recap their ingredients. The integration over replica spacetimes originates from \eqref{4}. As explained, the derivation and meaning of \eqref{4} was discussed in detail in chapter $2$ until Lemma $2.3$ of \cite{Averin:2024yeo}. We again refer to this discussion and especially to Fig. $4$ therein for a visualization of the functional integration in \eqref{4}. There, the integration is performed over time-slices $\Sigma_1 = \Sigma(B) \subseteq \Sigma$ along the time $\alpha \in [0, 2\pi].$ The integration over replica spacetimes $R_n$ has the same meaning except that the duration of time is now generalized to $\alpha \in [0, 2\pi n]$ with a possibly non-integer $n.$ 

The other important ingredient appearing in \eqref{21} and \eqref{22} is the generator $K$ (or equivalently $G$ related by \eqref{10}). As mentioned, $G$ generates the diffeomorphism corresponding to the evolution along the time $\alpha$ in Fig. $4$ (see also Fig. $5$) of \cite{Averin:2024yeo}. This requirement fixes $G$ uniquely up to a constant shift $G \to G+c.$ Since we will need the explicit form of the generators in what follows, we state them here. The actual derivation was done in \cite{Averin:2024yeo} in the proof of Theorem $3.1.$ Up to a constant shift (which we leave open for a moment), we have

\begin{equation}
K[\Phi] = 2\pi \oint_B X^{cd} \varepsilon_{cd}
\label{26}
\end{equation}

for a state $\Phi \in \Gamma^{(B)}.$ With $n = \dim \Sigma +1$ being the spacetime dimension, the $(n-2)$-form $X^{cd}$ can be explicitly obtained by a variational procedure from the Lagrangian of the given theory. $\varepsilon_{cd}$ is the binormal to $B$ oriented such that $\varepsilon_{cd} T^c X^d > 0.$ Hereby, $T^a$ is a future-directed timelike vector determining the orientation of $\Sigma$ and $X^a$ is a spacelike vector pointing outwards to $B$ determining the orientation of $B.$ The details of this definition are explicitly stated in Theorem $3.1$ in \cite{Averin:2024yeo}. There, the procedure of obtaining the explicit form of $K$ is also shown in an explicit example. In Example $3.1$ in \cite{Averin:2024yeo}, $K$ is shown for Einstein-Hilbert gravity with minimally coupled matter to be of the form

\begin{equation}
K = \frac{A}{4G}
\label{27}
\end{equation}

with $A$ denoting the area of $B$ in the state $\Phi \in \Gamma^{(B)}.$

With this ingredients, \eqref{21} and \eqref{22} are nearly fixed. We so far have specified $G$ only up to a constant shift $G \to G+c.$ From \eqref{21}, we see that such a shift $G \to G+c$ corresponds to a rescaling $Z_n \to k^n Z_n$ for a certain $k \in \mathbb{C}.$ One easily verifies that the modular entropy \eqref{23} is invariant under such a rescaling and hence is invariant under the shift $G \to G+c.$ 

While the modular entropy (and hence the entanglement entropy) is independent of the chosen shift, the shift (and hence the generators $G$ and $K$) are uniquely fixed by requiring \eqref{25}, i.e. the density matrix to be normalized. We insist on this requirement in what follows and will determine the associated normalization of $K$ in the following. 

With the requirement \eqref{25}, the entanglement entropy \eqref{5} can be obtained from \eqref{24} and \eqref{23} as

\begin{equation}
S = - \left. \frac{d}{dn} \right|_{n=1} \ln (Z_n).
\label{28}
\end{equation}

From \eqref{21} and \eqref{22}, we then have

\begin{equation}
S = \left. \frac{\int_{\Phi \in R_n} \operatorname{Vol}^{(B)}(\alpha) \frac{K[\Phi(\alpha=2\pi n)]}{\hbar} e^{-\frac{I_n [\Phi]}{\hbar}}}{\int_{\Phi \in R_n} \operatorname{Vol}^{(B)}(\alpha) e^{-\frac{I_n [\Phi]}{\hbar}}} \right|_{n=1}.
\label{29}
\end{equation}

In \eqref{29}, we have restored $\hbar.$ So far, the expression \eqref{29} for the entanglement entropy is exact. We see that \eqref{29} can be seen as an expectation value of the generator $K$ taken over all replica spacetimes $R_n.$ The natural way to evaluate \eqref{29} at least approximately is to ask for the dominant replica spacetimes contributing to the expectation value \eqref{29} in the $\hbar \to 0$ limit. Precisely this, we will do in the next section. 

Before doing so, we briefly fix the normalization of $K$ appearing in \eqref{29}. As mentioned, \eqref{29} is valid under the requirement \eqref{25} which determines $K$ uniquely. We can read off the normalization of $K$ directly from \eqref{29}. For any state \eqref{1}, we imagine making $B$ smaller until it collapses to the empty set. Under this procedure, the entanglement entropy must go to zero $S \to 0.$ We noticed that $K$ was given by \eqref{26} up to a constant shift. However, because $S \to 0$ under the described procedure, we learn from \eqref{29} that this constant shift must vanish. Hence, \eqref{29} is valid if $K$ is precisely given by \eqref{26} (which reduces to \eqref{27} in the appropriate special case).

\subsection{Leading Order}
\label{Kapitel 3.1}

We want to find a suitable approximation to the sum over replica spacetimes in \eqref{29}. As noted, the natural approximation is to ask which replica spacetimes dominate \eqref{29} in the $\hbar \to 0$ limit? 

Let $\Phi \in R_n$ be a replica spacetime. In order for the fluctuations around $\Phi$ to interfere constructively in \eqref{29}, the integrand has to be stationary. That is, both the action functional

\begin{equation}
\delta I_n = 0
\label{30}
\end{equation}

as well as the generator

\begin{equation}
\delta K[\Phi(\alpha=2\pi n)] = 0
\label{31}
\end{equation}

have to be stationary under variations around $\Phi.$\footnote{Note that (traces of) density matrices are often given by functional integrals of the general Schwinger-Keldysh form such as in equation $(2)$ of \cite{Averin:2024yeo}. While a given path gives a complex contribution to the functional integral, the hermiticity of the density matrix is ensured by the Schwinger-Keldysh form since the reversed path gives precisely the complex conjugate contribution. Saddle points of the functional integration always occur pairwise with one path in phase space following the flow generated by the Hamiltonian and the associated reversed path following the flow of the sign flipped Hamiltonian giving precisely the complex conjugate contribution. Furthermore, we note that the same reasoning applies to the case we are considering here.} Because the replica
spacetimes are required to be closed curves in $\Gamma^{(B)},$ the condition \eqref{30} requires the replica spacetime to follow the Hamiltonian evolution generated by $G$ in $\Gamma^{(B)}$ as follows from \eqref{20} (see also $(6)$ in \cite{Averin:2024its}). Such a stationary replica spacetime exists for any $n.$ In order to see this, we again refer to Fig. $4$ in \cite{Averin:2024yeo} which we, as explained, can interpret as the case $n=1.$ The stationarity of the evolution of $\Sigma$ along $t$ is equivalent to the stationarity of the evolution of $\Sigma(B)$ along $\alpha,$ i.e. it means the spacetime shown in Fig. $4$ in \cite{Averin:2024yeo} is a solution of the theory's equations of motion. A stationary evolution along $t$ exists and is obtained from the initial conditions on position coordinates at the times $t=\pm i T_E.$ It can then be translated to a stationary evolution along $\alpha,$ i.e. a replica spacetime for $n=1$ following the Hamiltonian flow along $G.$ Then, there is no obstacle in continuing this flow to arbitrary possibly non-integer $n$ differing from $1.$ 

While there is no problem in discussing replica spacetimes for non-integer $n,$ only for $n=1$ equivalence of the evolutions of Fig. $4$ in \cite{Averin:2024yeo} along the time $t$ and angle $\alpha$ exists. For non-integer $n$ the spacetime shown in Fig. $4$ in \cite{Averin:2024yeo} could be interpreted as possessing a conical deficit at $B.$

Altogether, we learn that for arbitrary $n,$ a replica spacetime $\Phi \in R_n$ can be chosen to satisfy condition \eqref{30}. However, to dominate \eqref{29} in the $\hbar \to 0$ limit, the extremization condition \eqref{31} has to be additionally fulfilled. We hence restrict to stationary replica spacetimes $\Phi$ (i.e. \eqref{30} holds) extremizing $K[\Phi(\alpha=2\pi n)].$

This means $\Phi(\alpha)$ follows the Hamiltonian evolution along $\alpha$ generated by $G$ which, as mentioned, is a diffeomorphism behaving as a rotation with angle $\alpha$ near $B.$ Since $K$ only depends on $\Phi(\alpha)$ through quantities on $B$ (see \eqref{26}), $K[\Phi(\alpha)]$ (or $G[\Phi(\alpha)]$ according to \eqref{10}) is constant in $\alpha$ for fixed $n.$ A look at \eqref{20} then reveals that the considered replica spacetimes are
exponentially suppressed in the nominator and denominator of \eqref{29} by a factor given by their extremized value of $\frac{1}{\hbar} K[\Phi(\alpha=2\pi n)].$ The dominant replica spacetime in the fraction in \eqref{29} is then the one with the minimal value of the extremized $\frac{1}{\hbar} K[\Phi(\alpha=2\pi n)].$ 

In total, we then obtain for the fraction in \eqref{29} in the $\hbar \to 0$ limit 

\begin{equation}
-\frac{d}{dn} \ln (Z_n) \xrightarrow[\hbar \to 0]{} \underset{\Phi \in R_n}{\operatorname{min} \operatorname{Extr}} \left\{ \frac{1}{\hbar} K[\Phi(\alpha=2\pi n)] \bigg| \delta I_n = 0 \right\}.
\label{32}
\end{equation}

With this, we obtain from \eqref{28} the following general prescription for the entanglement entropy \eqref{5} to leading order in the $\hbar \to 0$ limit

\begin{equation}
S \xrightarrow[\hbar \to 0]{} \lim_{n \to 1} \underset{\Phi \in R_n}{\operatorname{min} \operatorname{Extr}} \left\{ \frac{1}{\hbar} K[\Phi(\alpha=2\pi n)] \bigg| \delta I_n = 0 \right\}.
\label{33}
\end{equation}

The expression \eqref{33} gives a practical way to evaluate the summation in \eqref{29} at least to leading order in $\hbar.$ In order to concretely apply \eqref{33} for the computation of entanglement entropies, we follow the instructions used in the derivation of \eqref{32}. For convenience, we briefly summarize them in the following prescription:

\begin{figure}[h!]
\centering
  \includegraphics[trim = 0mm 90mm 0mm 30mm, clip, width=0.7\linewidth]{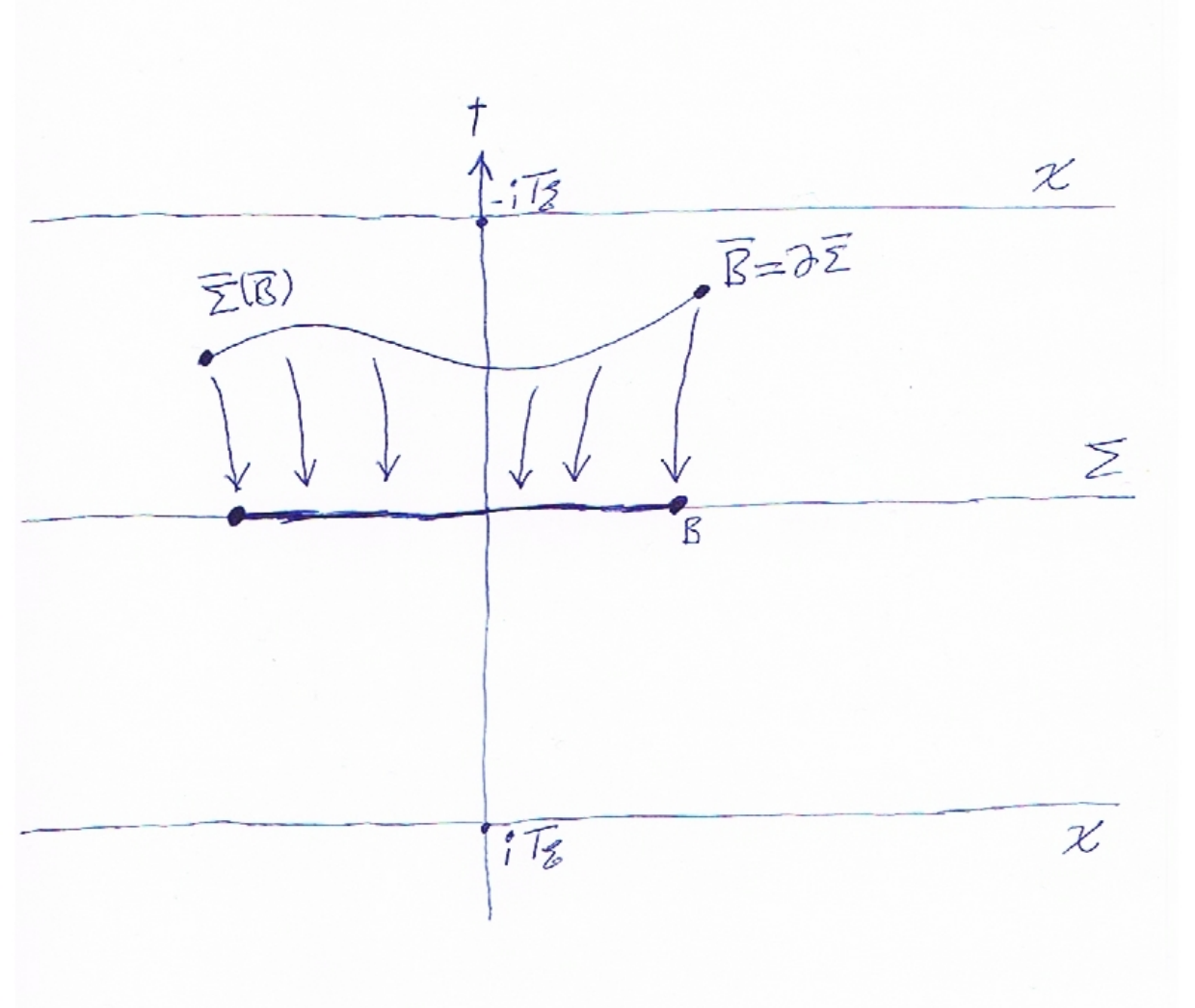}
  \caption{Visualization of the prescription to evaluate \eqref{33}.}
	\label{Fig. 1}
\end{figure}

\begin{description}
\item[Step 1:]We find a spacetime foliated by $\Sigma$ along $t$ solving the theory's equations of motion from the initial conditions at $t=\pm iT_E.$ The procedure is visualized in Fig. \ref{Fig. 1}.
\item[Step 2:]We translate the spacetime found in the first step and shown in Fig. \ref{Fig. 1} to a stationary replica spacetime $\bar{\Phi}(\alpha)$ foliating $\Sigma(B)$ along $\alpha.$ 

This translation is not unique. The theory under consideration is a gauge theory (at least diffeomorphism invariant). Hence, we can perform gauge transformations in the interior of the spacetime in Fig. \ref{Fig. 1} leaving the evolution of $\Sigma$ along $t$ unchanged. However, the evolution of $\Sigma(B)$ along $\alpha$ in general does change leading to an inequivalent stationary replica spacetime. For instance, in Fig. \ref{Fig. 1} a new replica spacetime is constructed by dragging the surface $\bar{\Sigma}(\bar{B})$ bounded by $\bar{B}$ to $\Sigma(B)$ with a diffeomorphism represented by the arrows. Therefore, all codimension-2 surfaces $\bar{B}$ in the spacetime in Fig. \ref{Fig. 1} that can be dragged to $B$ in the mentioned way lead to replica spacetimes which are generally different elements of the set in \eqref{33}.
\item[Step 3:]The second step has constructed the set in \eqref{33} over which we have to perform the extremization and minimization procedure as required by \eqref{33}.
\end{description}

We close this section with some comments. 

It is remarkable that, according to the described prescription, \eqref{33} allows one to read off quantum mechanical entanglement entropies from classical geometries. This was already noticed in the context of the bulk-boundary correspondence in \cite{Ryu:2006bv,Ryu:2006ef} where it is known as the Ryu-Takayanagi formula. Our derived prescription \eqref{33} generalizes the Ryu-Takayanagi formula (and existing variants). While we here see that \eqref{33} contains all of the ingredients of the Ryu-Takayanagi formula, we postpone a detailed comparison to the next chapter. 

As explained, the extremizing and minimizing replica spacetime in \eqref{33} corresponds to a deformation $\bar{B}$ of the codimension-2 surface $B$ with some suited diffeomorphism in the spacetime in Fig. \ref{Fig. 1}. In the context of the Ryu-Takayanagi formula, the surface $\bar{B}$ is hence known as the \emph{extremal surface}. 

Note that the deformations $\bar{B}$ and hence the extremal surface are allowed in Fig. \ref{Fig. 1} also to lie in regions of real time $t.$\footnote{Whether a gauge transformation like a specific deformation $\bar{B}$ of $B$ is allowed in Fig. \ref{Fig. 1} to provide a legal replica spacetime in the set in \eqref{33} is determined by the definition of the possifold flow $\partial \Sigma \to B.$ It determines which canonical degrees of freedom are contained in $B$ or its complement. This information is encoded in the canonical form $\Theta$ defining the theory under consideration. We refer to \cite{Averin:2024its} where this issue was discussed in detail in the definition of possifold flows for field theories.}

The limes is essential in the prescription \eqref{33}. In general, it does not commute with the extremization procedure. This was already noted in \cite{Dong:2013qoa,Dong:2017xht}. 

Altogether, we have found in this section the prescription \eqref{33} which evaluates the entanglement entropy \eqref{29} to leading order in $\hbar.$ This leading order contribution can hereby be inferred from a classical geometry by searching for an extremal surface which extremizes the functional $K$ as formalized in \eqref{33}. According to \eqref{27}, for the case of Einstein-Hilbert gravity, the extremal surface has to extremize its surface area. 

In the next section, we want to evaluate the functional integral expression of the entanglement entropy \eqref{29} beyond leading order. 

\subsection{Beyond Leading Order}
\label{Kapitel 3.2}

As in the last section, we want to find a practical way for evaluating expression \eqref{29} for the entanglement entropy.

However, in this section, our aim is to obtain an expression valid beyond leading order in $\hbar.$ 

For arbitrary $n \in \mathbb{R},$ we choose a replica spacetime $\bar{\Phi} \in R_n.$ Using \eqref{23} and \eqref{24}, we can write 

\begin{equation} \label{34}
\begin{split}
S &= \lim_{n \to 1} \left(1-n\frac{d}{dn} \right) \ln (Z_n) \\
&= \lim_{n \to 1} \left( \frac{1}{\hbar}K[\bar{\Phi}(\alpha=2\pi n)] + \left(1-n\frac{d}{dn} \right) \ln (Z_n) -\frac{1}{\hbar}K[\bar{\Phi}(\alpha=2\pi n)] \right).
\end{split}
\end{equation}  

With the notation being clear in a moment, we define

\begin{equation}
S_{bulk} [\bar{\Phi}] := \left(1-n\frac{d}{dn} \right) \ln (Z_n) -\frac{1}{\hbar}K[\bar{\Phi}(\alpha=2\pi n)].
\label{35}
\end{equation}

Since \eqref{34} does not depend on $\bar{\Phi},$ we can obviously write

\begin{equation}
S = \lim_{n \to 1} \underset{\bar{\Phi} \in R_n}{\operatorname{min} \operatorname{Extr}} \left\{ \frac{1}{\hbar}K[\bar{\Phi}(\alpha=2\pi n)] + S_{bulk} [\bar{\Phi}] \right\}.
\label{36}
\end{equation}

We define $\bar{Z}_n$ as the function in $n$ solving the differential equation 

\begin{equation}
\left(1-n\frac{d}{dn} \right) \ln (\bar{Z}_n) = \frac{1}{\hbar}K[\bar{\Phi}(\alpha=2\pi n)]
\label{37}
\end{equation}

under the initial condition $\bar{Z}_{n=1} = 1.$ If $\bar{\Phi} \in R_n$ is a stationary replica spacetime fulfilling the conditions of \eqref{32}, then $\bar{Z}_n$ is the partition function in the vicinity of $n=1$ to leading order in $\hbar$ as follows from \eqref{32}. In that case, the modular entropy \eqref{35} 

\begin{equation}
S_{bulk} [\bar{\Phi}] = \left(1-n\frac{d}{dn}\right) \ln \left( \frac{Z_n}{\bar{Z}_n} \right)
\label{38}
\end{equation}

has a simple meaning. $\frac{Z_n}{\bar{Z}_n}$ can, to leading order in $\hbar,$ be interpreted as the partition function of the fluctuations around the background replica spacetime $\bar{\Phi} \in R_n.$ $S_{bulk}[\bar{\Phi}]$ is then the associated modular entropy of these fluctuations around the background (or bulk) spacetime $\bar{\Phi}.$ This explains the notation introduced in \eqref{35}. 

We are now able to formulate a practical prescription for the evaluation of the entanglement entropy \eqref{29} to arbitrary order in $\hbar.$ First, we compute $S_{bulk}[\bar{\Phi}]$ as defined in \eqref{38} to arbitrary order in a sensible perturbative expansion for any replica spacetime $\bar{\Phi} \in R_n.$ To the chosen order in perturbation theory, the entanglement entropy is obtained by the extremization procedure \eqref{36}. This extremization procedure proceeds precisely in the same manner as in the extremal surface prescription \eqref{33} of the last section and all comments essentially apply unchanged. There are only two essential differences. Contrary to \eqref{33}, we have to extremize over all $\bar{\Phi} \in R_n$ in \eqref{36} which need not be solutions of the equations of motion. Furthermore, in \eqref{36} we are extremizing with respect to a different functional than in \eqref{33}. In the context of the Ryu-Takayanagi formula, the searched codimension-2 surface $\bar{B}$ is known in this quantum corrected prescription as the \emph{quantum extremal surface.}

While the extremal surface contribution \eqref{33} can be inferred from a classical geometry, the leading order corrections due to the quantum extremal surface prescription \eqref{36} can be inferred from a semiclassical geometry. In fact, this can be done in the very practical manner described in the three steps of the last section. The first two steps are identical (see accordingly Fig. \ref{Fig. 1} for a visualization). The third step applies similarly except that the extremization procedure in \eqref{33} has now to be performed with respect to

\begin{equation}
\frac{1}{\hbar}K[\bar{\Phi}(\alpha=2\pi n)] + S_{bulk} [\bar{\Phi}].
\label{39}
\end{equation}

Since $\bar{\Phi} \in R_n$ is a replica spacetime solving the equations of motion, $S_{bulk} [\bar{\Phi}]$ is to leading order given by the modular entropy of the fluctuations around the background $\bar{\Phi}$ as explained below \eqref{38}. Hence, $S_{bulk} [\bar{\Phi}]$ can be inferred from a quantum field theory on curved spacetime calculation. In other words, the quantum extremal surface $\bar{B}$ is found in Fig. \ref{Fig. 1} by deforming $B$ so as to extremize the quantity \eqref{39} which to this order also involves quantum fluctuations around $\bar{\Phi}.$ The entanglement entropy can thus be obtained from the semiclassical geometry around the spacetime shown in Fig. \ref{Fig. 1}.

Having given a practical way to compute entanglement entropies \eqref{29} in diffeomorphism invariant field theories in this chapter, we answered the question raised at the beginning of this chapter. In the next chapter, we want to compare our results with the existing ones. Previous statements about the entanglement entropy in gravitational theories are formulated in the context of the bulk-boundary correspondence. We therefore need to specialize to this case. 

\section{Comparison with the Bulk-Boundary Correspondence}
\label{Kapitel 4}

In chapter \ref{Kapitel 2} we have derived a functional integral expression for the entanglement entropy in diffeomorphism invariant field theories. A practical prescription for the evaluation of the entanglement entropy was consequently obtained in chapter \ref{Kapitel 3}. In the special case of the bulk-boundary correspondence, that is, if the gravitational theory under consideration possesses a holographic dual, prescriptions for the entanglement entropy already exist. In this chapter, we want to compare our general prescriptions of chapter \ref{Kapitel 2} and \ref{Kapitel 3} with the existing results. We will see that the existing results emerge as special cases of the discussion in chapter \ref{Kapitel 2} and \ref{Kapitel 3}.

To begin our comparison, we need to adjust the prescriptions of chapter \ref{Kapitel 2} and \ref{Kapitel 3} to the special case of the bulk-boundary correspondence. For this, there is only a little adjustment to be taken. It concerns the choice of the codimension-2 surface $B$ defining the possifold flow in the situation of chapters \ref{Kapitel 2} and \ref{Kapitel 3}. In the context of the bulk-boundary correspondence, $B = X \cup A$ is required to be a disjoint union of a subset $A \subseteq \partial \Sigma$ of the spacetime boundary and a subset $X \subseteq \Sigma$ lying in the bulk. The situation is illustrated in Fig. \ref{Fig. 2}. 

\begin{figure}[h!]
\centering
  \includegraphics[trim = 0mm 90mm 0mm 60mm, clip, width=0.7\linewidth]{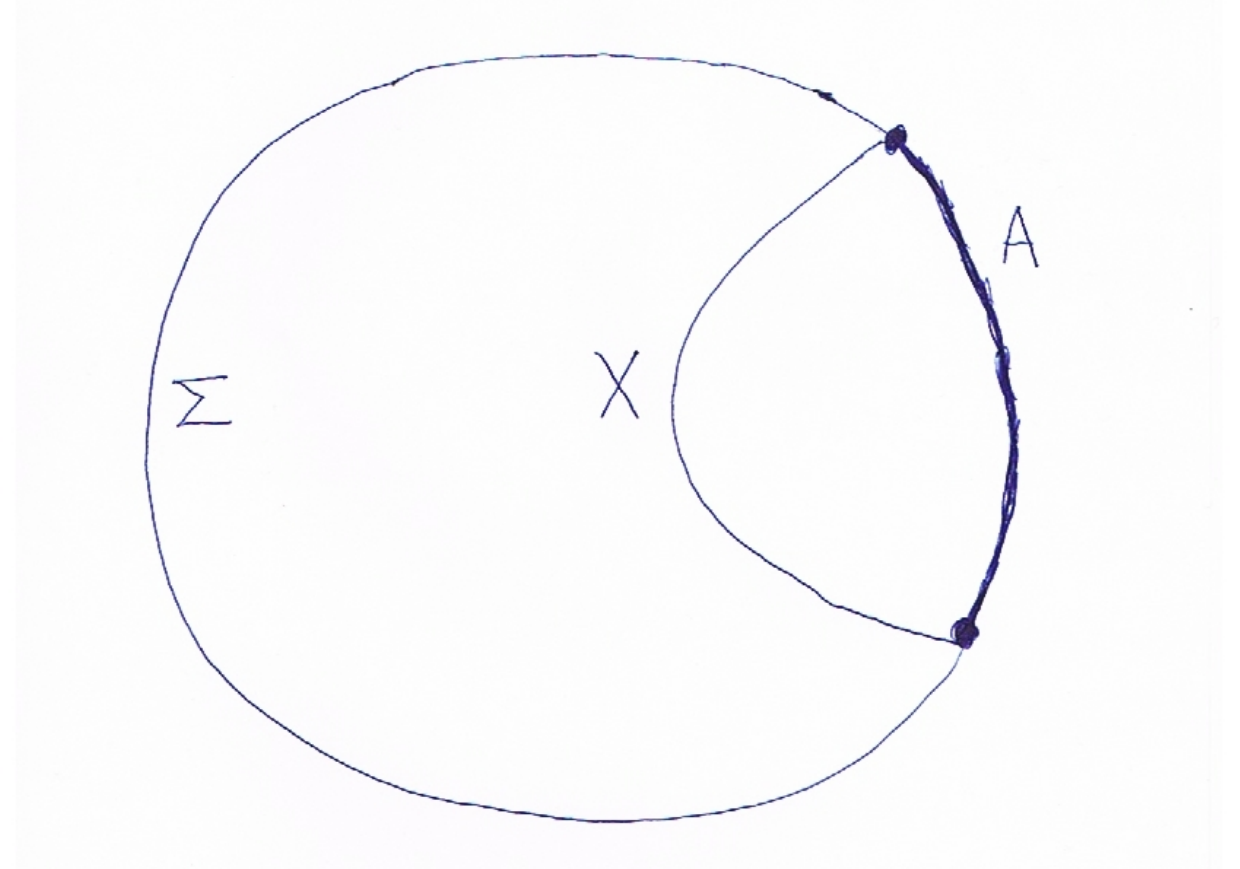}
  \caption{$\Sigma$ is represented by the disk with the circle being its boundary. The possifold flow is defined by the surface $B= X \cup A.$ It bounds the region $\Sigma(B)$ between $X$ and $A.$}
	\label{Fig. 2}
\end{figure} 

In this situation, the only adjustment to be made in the prescriptions of chapter \ref{Kapitel 2} and \ref{Kapitel 3} is to carefully compute the generator $K$ and to adopt \eqref{26} accordingly. In order to do so, we need to repeat the derivation of the generator $K$ or equivalently $G$ in \eqref{10} as was done in \cite{Averin:2024yeo}. We follow this derivation and refer to \cite{Averin:2024yeo} for details. The role of the generator $G$ was to implement a certain diffeomorphism $\xi.$ With the same notation as in \eqref{26}, $G$ is according to equation $(47)$ of \cite{Averin:2024yeo} up to a constant given by

\begin{equation}
G = \oint_B X^{cd}(\Phi) \nabla_{[c} \xi_{d]}.
\label{40}
\end{equation}

The contribution of the part $X$ of $B$ to \eqref{40} is unchanged compared to the derivation in \cite{Averin:2024yeo}. Here, as already mentioned several times and discussed in detail in \cite{Averin:2024yeo}, the diffeomorphism $\xi$ near $X$ looks like a rotation in the plane normal to $X.$ Accordingly, the contribution of $X$ to $G$ and then to $K$ is unchanged and hence of the same form as in \eqref{26}. However, since the part $A \subseteq \partial \Sigma$ is located inside the boundary $\partial \Sigma,$ the diffeomorphism $\xi$ vanishes near $A.$ Hence, there is due to \eqref{40} no contribution to $G$ and then $K$ from $A.$ In total, this means that \eqref{26} is altered to

\begin{equation}
K[\Phi] = 2\pi \int_X X^{cd} \varepsilon_{cd}
\label{41}
\end{equation}

for a state $\Phi \in \Gamma^{(B)}.$ 

Accordingly, following Example $3.1$ of \cite{Averin:2024yeo}, \eqref{27} is altered to

\begin{equation}
K=\frac{A(X)}{4G}
\label{42}
\end{equation}

where $A(X)$ denotes the area of $X$ in a particular state. 

Applying the prescriptions together with the explanations of the last chapter \ref{Kapitel 3} with these adjustments, we obtain the existing results as special cases. For Einstein-Hilbert gravity with minimally coupled matter, where $K$ has the form \eqref{42}, the prescription \eqref{33} matches precisely the covariant holographic entanglement entropy proposal \cite{Hubeny:2007xt} which is a generalization of the original Ryu-Takayanagi conjecture \cite{Ryu:2006bv,Ryu:2006ef} (see \cite{Rangamani:2016dms} for a review). For higher-derivative gravity, where $K$ is generally given by \eqref{41}, the prescription \eqref{33} matches formula $(2.14)$ of \cite{Dong:2017xht} and hence the generalization of the Ryu-Takayanagi formula in \cite{Dong:2013qoa}. The prescription \eqref{36} to include corrections beyond the leading order matches the quantum extremal surface prescription of \cite{Faulkner:2013ana,Engelhardt:2014gca,Dong:2017xht}. 

Altogether, we see that our prescriptions for the entanglement entropy derived in chapters \ref{Kapitel 2} and \ref{Kapitel 3} reproduce the existing results in the literature. Before concluding this chapter, we want to compare our derivation with the ones of the existing results.

A derivation of the Ryu-Takayanagi conjecture was first presented in \cite{Lewkowycz:2013nqa}. A review of this derivation together with a proof of the more general covariant holographic entanglement entropy proposal is given in \cite{Rangamani:2016dms}. We briefly summarize the main points of this derivation and compare them to our derivation presented here. 

The starting point is to consider a gravitational theory which is assumed to possess a boundary dual in the situation as in Fig. \ref{Fig. 2}. The boundary theory is then defined on $\partial \Sigma$ and one is interested in entanglement entropies with respect to the part $A \subseteq \partial \Sigma$ of the boundary. In the boundary theory this is naturally done by defining $Z_n= \operatorname{tr} (\rho^n)$ with $\rho$ being the density matrix of the considered state with respect to $A$ and then using the replica trick. The assumption of the bulk-boundary correspondence states that $Z_n$ can equally well be obtained in the gravitational bulk theory. For integer $n,$ the partition function is then approximated as $Z_n= e^{-I_n}$ with $I_n$ being the action of the spacetime $M_n$ representing a dominant saddle point of the gravitational bulk theory. The spacetime $M_n$ is hence obtained as a solution to the bulk equations of motion with boundary conditions as set in the definition of $Z_n= \operatorname{tr} (\rho^n)$ in the boundary theory. In the boundary theory, $Z_n$ is viewed as an integration over a manifold known as the $n$-fold branched cover. The important point is that this $n$-fold branched cover has a $\mathbb{Z}_n$ replica symmetry reflecting cyclic permutation symmetry in the definition of $Z_n= \operatorname{tr} (\rho^n).$ It is then assumed that the spacetime $M_n$ inherits this $\mathbb{Z}_n$ symmetry with a codimension-2 surface $X$ being the associated fixed point as represented in Fig. \ref{Fig. 2}. The essence of the derivation in \cite{Lewkowycz:2013nqa,Rangamani:2016dms} is then to analytically continue the situation to non-integer $n,$ where the spacetime $M_n$ then has a conical singularity at $X.$ The main part is to compute $\frac{dI_n}{dn}$ needed for processing the replica trick which is shown to receive only contributions from the singular surface $X$ and then to give rise to the expected expression for the entanglement entropy. The more general derivations in \cite{Dong:2013qoa,Dong:2017xht} essentially proceed in the same way by considering more general gravitational theories and/or quantum corrections. 

The role of $M_n$ in Step $2$ of the prescription in chapter \ref{Kapitel 3.1} is played by the replica spacetime $\bar{\Phi}(\alpha) \in R_n.$ Our derivation in chapter \ref{Kapitel 3.1} shows its existence where we explained the meaning of the R\'{e}nyi-index $n$ and the associated $\mathbb{Z}_n$ symmetry for integer $n$ as well as the conical singularity at the $\mathbb{Z}_n$ fixed point $B$ for non-integer $n.$ In the derivations reviewed in the previous paragraph, the derivative $\frac{dI_n}{dn}$ has to be determined by carefully taking into account the behavior of spacetime fields near the conical singularity (see the formalism of squashed cones in \cite{Dong:2013qoa}). In our derivation presented in chapter \ref{Kapitel 2} and \ref{Kapitel 3}, this derivative is naturally obtained in Lemma \ref{Lemma 2.1}. There, this derivative is obtained in a simple manner because of the natural functional integral from \eqref{4} of the density matrix. This in turn is a consequence of the commutativity of summation in the functional integration interplayed with the diffeomorphism invariance of the field theory under consideration. 

\section{Discussion and Outlook}
\label{Kapitel 5}

What do we have achieved?

We have derived a functional integral expression for the entanglement entropy \eqref{29} for diffeomorphism invariant field theories. In chapter \ref{Kapitel 3}, we have stated and proven a practical prescription for the evaluation of \eqref{29} to leading order in $\hbar$ and beyond. These results fulfill the task asked for in the introduction.

Furthermore, we have shown in chapter \ref{Kapitel 4} that our results reproduce the Ryu-Takayanagi formula and its various generalizations in the context of the bulk-boundary correspondence. Our results provide hereby an independent derivation which clearly unravels the physical origin as well as the mathematics behind the Ryu-Takayanagi type formulas. As explained in detail, the intimate interplay between the covariance of the functional integral and diffeomorphism invariance is the reason for their existence. 

In general, as with any quantum theory following the framework of \cite{Averin:2024its}, we view the quantum theory of gravity as being given by a weighted sum over paths in phase space. Manipulations of this sum exploiting its commutativity together with diffeomorphism invariance lead maybe surprisingly to formal statements and proofs of results that were otherwise only either conjectured on a heuristic basis or stated in the context of the bulk-boundary correspondence. An example for the former is the gravitational entropy bound which is generally motivated based on black hole thermodynamics. A formal statement and proof was given in \cite{Averin:2024yeo} made possible by the described framework. An example for the latter is the gravitational entanglement entropy derived here by the same methods. 

It is notable that highly expected non-perturbative statements about quantum gravity (expected on different grounds) can be formally stated and derived within the explained framework. This is highly supportive for the highly non-trivial viewpoint on the quantum theory of gravity as being given by a sum over paths in phase space. 

An intrinsic feature of this viewpoint is the contact to phase space. In fact, the gravitational entropy bound is a statement on the geometry of phase space. At the same time, it is this contact to phase space that allowed us to generalize the Ryu-Takayanagi type formulas to general gravitational theories not necessarily possessing a holographic dual. This provides the additional advantage in our approach that it singles out the concrete part of phase space that is responsible for a particular entanglement entropy that is being evaluated in a concrete example. 

This brings us to our last point. We have here given no concrete examples. In the case of the bulk-boundary correspondence, our prescription on the evaluation of entanglement entropies reduces to the known Ryu-Takayanagi type formulas. A lot of examples on this exist in the literature (see, for instance, \cite{Rangamani:2016dms}). The more interesting case is of course to consider a situation in a gravitational theory without known holographic dual. Such a situation would be a black hole in asymptotically flat space. As emphasized, our approach makes direct contact with the phase space relevant for entanglement entropy. The hope is hence that our results on entanglement entropy make a statement about the part of phase space responsible for black hole entropy and microstates.

The analysis of this example deserves its own space and we will present it in a different place.      

\section*{Acknowledgements}
We thank Alexander Gußmann for many discussions on this and other topics in physics.

\end{document}